\documentclass{llncs} 
\usepackage{amssymb,color}
\usepackage{verbatim}
\usepackage{comment}
\usepackage{xspace}
\usepackage{multirow}
\usepackage{thmtools,thm-restate}
\newcommand{\newcomment}[1]{}
 
\newcommand{\FI}{\ensuremath{\mathrm{FI}}\xspace}

\newcommand{\E}{\mathbb{E}}
\newcommand{\JS}{\mathrm{JS}}
\newcommand{\Ham}{\mathrm{Ham}}
\newcommand{\Sim}{\mathrm{Sim}}
\newcommand{\D}{\mathcal{D}}

\newcommand{\R}{\mathbb{R}}

\usepackage{amsmath}
\usepackage[boxed,linesnumbered]{algorithm2e}
\SetKwInput{KwInput}{Input}
\SetKwInput{KwOutput}{Output}

\declaretheorem[name=Theorem]{thm}
\declaretheorem[name=Lemma,numberlike=thm]{lem}
\declaretheorem[name=Corollary,numberlike=thm]{cor}

\begin{document}

\title{Frequent-Itemset Mining using Locality-Sensitive Hashing}

\author{
Debajyoti Bera \inst{1}
  \and
Rameshwar Pratap\inst{2}
}

%

\institute{Indraprastha Institute of Information Technology-Delhi (IIIT-D), India\\
	\email{dbera@iiitd.ac.in} 
\and TCS Innovation Labs, India\\
\email{rameshwar.pratap@gmail.com}} 
   
\date{}
\maketitle

\begin{abstract}
The Apriori algorithm   is a classical algorithm
for the frequent itemset mining problem. 
A significant bottleneck in Apriori is the number of I/O operation involved, 
and the number of candidates it generates. 
We investigate the role of LSH techniques to overcome  these problems, without adding  much computational overhead. 
We propose randomized variations of Apriori that are based on asymmetric LSH
 defined over Hamming distance and Jaccard similarity.
%
\end{abstract}

\section{Introduction}\label{sec:intro}
Mining \emph{frequent itemsets} in a transactions database appeared first in
the context of analyzing supermarket transaction data for discovering
association rules~\cite{AgrawalS94,AgrawalIS93}, however this problem has, since then, found
applications in diverse domains like finding
correlations~\cite{SilversteinBM98}, finding episodes~\cite{MannilaTV97}, 
clustering~\cite{WangWYY02}. 
Mathematically, each transaction can be regarded
as a subset of the items (``itemset'') those that present in the transaction. Given a
database $\D$ of such transactions and a support threshold $\theta\in(0,1)$, the
primary objective of frequent itemset mining is to identify $\theta$-frequent itemsets
(denoted by \FI, these are subsets of items
that appear in at least $\theta$-fraction of transactions).

Computing \FI is a challenging problem of data mining.
The question of deciding if there
exists any \FI~with $k$ items is known to be NP-complete~\cite{GunopulosKMSTS03}
(by relating it to the existence of bi-cliques of size
$k$ in a given bipartite graph) but on a more practical note, simply  checking support of
any itemset requires reading the transaction database -- something that is
computationally expensive since they are usually of an extremely large size. The state-of-the-art approaches try to reduce the number of candidates, or not
generate candidates at all. The best known approach in the former line
of work is the celebrated Apriori algorithm~\cite{AgrawalS94}.

Apriori is based on the {\em anti-monotonicity property of \em
partially-ordered sets} which says that no superset of an infrequent itemset can
be frequent. This algorithm works in a bottom-up fashion by generating itemsets of size
$l$ in level $l$, starting at the first level. After finding frequent
itemsets at level $l$ they are joined pairwise to generate $l+1$-sized
{\em candidate itemsets}; \FI are identified among the candidates
by computing their support explicitly from the data. 
The algorithm terminates when no more candidates are generated.
Broadly, there are two downsides to this simple but effective algorithm. 
The first one is that the algorithm has to compute support 
\footnote{Note that computing support  is an I/O intensive operation and involves reading every transaction.}
of every itemset 
in the candidate, even the ones that are highly infrequent. 
Secondly, if an itemset is infrequent, but all its subsets are frequent,
Apriori doesn't have any easy way of detecting this without reading 
every transaction of the candidates.

A natural place to look for fast algorithms over large data are
randomized techniques; so we investigated if LSH
could be be of any help. An earlier work by Cohen {\it et al.}~\cite{Cohen01}
was also motivated by the same idea but worked on a different problem
(see Section~\ref{subsec:relwork}). LSH is explained in
Section~\ref{sec:background}, but roughly, it is a randomized hashing technique
which allows efficient retrieval of approximately ``similar'' elements (here,
itemsets).

\subsection{Our contribution}

In this work, we propose LSH-Apriori -- a basket of three explicit variations of
Apriori that uses LSH for
computing $\FI$. LSH-Apriori handles both the above mentioned drawbacks of the Apriori algorithm.
First, LSH-Apriori significantly cuts down on the number of infrequent candidates that 
are generated, and further due to its  dimensionality reduction property 
 saves on reading every transaction;   secondly, 
 LSH-Apriori  could efficiently filter our those infrequent itemset without looking every candidate.
The first two variations essentially reduce computing \FI to the approximate nearest neighbor
(cNN) problem for Hamming distance and Jaccard similarity. Both these approaches
 can drastically reduce the number of false candidates without much
overhead, 
but has a non-zero probability of error in the sense that some frequent itemset could be missed by the algorithm. 
Then we present a third variation which also maps \FI to elements in the Hamming
space but avoids the problem of these false negatives 
incurring a little cost of time and space complexity. Our techniques are based on
asymmetric LSH~\cite{Shrivastava015} and LSH with one-sided error~\cite{Pagh16} which are 
proposed very recently.

 
\subsection{Related work}\label{subsec:relwork}

There are a few hash based heuristic to compute $\FI$ which outperform the Apriori algorithm and PCY~\cite{PCY95} is 
  one of the most notable among them. 
PCY focuses on using hashing to efficiently utilize the main memory over each pass of
the database. However, our objective and approach both are fundamentally
different from that of PCY.

The work that comes closest to our work is by Cohen \emph{et al.}~\cite{Cohen01}. They developed a family of algorithms for finding interesting associations 
in a transaction database, also using LSH techniques.
However, they specifically wanted to avoid any kind of filtering of itemsets
based on itemset support. On the other hand, our problem is the vanilla frequent
itemset mining which requires filtering itemsets satisfying a given minimum
support threshold.


\subsection{Organization of the paper}
In Section~\ref{sec:background}, we introduce the relevant concepts and give
an overview of the problem. In Section~\ref{sec:LSH-Apriori}, we build up the concept of  LSH-Apriori
which is required to develop  our algorithms. 
In Section~\ref{sec:FIviaLSH}, we present three specific variations of
LSH-Apriori for computing \FI. Algorithms of
Subsections~\ref{subsection:HammingLSH} and~\ref{subsection:minHashLSH} are
based on Hamming LSH and Minhashing, respectively.
In Subsection~\ref{subsection:coveringLSH}, we present another approach based on CoveringLSH which overcomes the problem of 
  producing false negatives. In Section~\ref{sec:conclusion}, we summarize the whole discussion.

\section{Background}\label{sec:background}
\begin{tabular}{|c|l||c|l|}
\hline
\multicolumn{4}{|c|}{\bf Notations}\\
\hline
$\D$ & Database of transactions: $\{t_1,\ldots, t_n\}$  & $n$ & Number of transactions\\
\hline
$\D_l$ & \FI~of level-$l$: $\{I_1, \ldots I_{m_l} \}$ & $\theta$ & Support threshold, $\theta\in (0,1)$\\
\hline
$\alpha_l$ & Maximum support of any item  in $\D_l$ & $m$ & Number of items \\
\hline
$\varepsilon$ & Error tolerance in LSH, $\varepsilon\in (0,1)$ &  $m_l$ & Number of~\FI~of size $l$ \\
\hline
 $\delta$ & Probability of error in LSH, $\delta \in (0,1)$ & $|v|$ & Number of $1'$s in  $v$ \\
\hline
\end{tabular}\\

The input to the classical frequent itemset mining problem is a database $\D$ of
$n$ transactions
$\{T_1, \ldots, T_n\}$ over $m$ items $\{i_1, \ldots, i_m\}$ and a support
threshold $\theta\in(0,1)$. Each transaction, in turn, is a subset of those
items. Support of itemset ${I \subseteq \{i_1, \ldots, i_m\}}$ is the number
of transactions that contain $I$. The objective of the problem is to determine
{\em every itemset} with support at least $\theta n$. We will often identify an itemset $I$
with its transaction vector $\langle I[1], I[2], \ldots, I[n]\rangle$ where
$I[j]$ is 1 if $I$ is contained in $T_j$ and 0 otherwise. An equivalent way to
formulate the objective is to find itemsets with at least $\theta n$ 1's in their
transaction vectors. It will be useful to view $\D$ as a set of
$m$ transaction vectors, one for every item.

 \subsection{Locality Sensitive Hashing}
We first briefly explain the concept of locality sensitive hashing (LSH).
 \begin{definition}[Locality sensitive hashing~\cite{IndykM98}]\label{definition:LSH}
 Let $S$ be a set of $m$ vectors in $\R^n$, and $U$ be the hashing
universe. Then, a family $\mathcal{H}$ of functions from $S$ to $U$
is called as $(S_0, (1-\varepsilon) S_0, p_1, p_2)$-sensitive (with $\varepsilon\in(0, 1]$ and $p_1 > p_2$) for the similarity measure $Sim(.,.)$ if for any
$x, y \in S$:
\begin{itemize}
    \item if $\Sim(x,y)\geq S_0$, then $\displaystyle \Pr_{h \in \mathcal{H}}[h(x)=h(y)]\geq p_1$,
    \item if $\Sim(x,y)\leq (1-\varepsilon) S_0$, then $\displaystyle \Pr_{h \in \mathcal{H}}[h(x)=h(y)]\leq p_2.$
\end{itemize}
\end{definition}
%

Not all similarity measures have a corresponding LSH. However, the following
well-known result gives a sufficient condition for
existence of LSH for any $Sim$.

\begin{lem}\label{lem:monotonicLSH}
If $\Phi$ is a strict
monotonic function and a family of hash
function $\mathcal{H}$ satisfies $\Pr_{h\in \mathcal{H}}
[h(x)=h(y))=\Phi(Sim(x,y)]$ for some $Sim:\R^n \times \R^n \to \{0,1\}$, 
then the conditions of Definition~\ref{definition:LSH} are  true for $Sim$ for any
$\varepsilon\in (0, 1)$.
\end{lem}

 The similarity measures that are of our interest   are {\em Hamming} and {\em
Jaccard} over binary vectors. Let $|x|$ denote
the Hamming weight of a binary vector $x$. Then, for  vectors $x$ and $y$
of length $n$, Hamming distance is defined as $\Ham(x,y) =|x \oplus y|$, where $x \oplus
y$ denotes a vector that is element-wise Boolean XOR of $x$ and $y$. Jaccard
similarity is defined as 
${\langle x,y \rangle}/{|x \vee y|}$, 
where
$\langle x, y \rangle$ indicates inner product, and $x \vee y$ indicates
element-wise Boolean OR of $x$ and $y$. LSH 
for these similarity measures are  simple and well-known~\cite{IndykM98,GionisIM99,BroderCFM00}. We recall them
below; here $I$ is some subset of $\{1, \ldots, n\}$ (or, $n$-length transaction
vector).

\begin{definition}[Hash function for Hamming distance]
 For any  particular bit position $i,$~
we  define the function $h_i(I):=I[i]$.
We will use hash functions of the form $g_J(I) = \langle h_{j_1}(I),
h_{j_2}(I),\ldots, h_{j_k}(I)\rangle$, where $J=\{j_1, \ldots, j_k \}$ is a
subset of $\{1, \ldots, n\}$ and the hash values are binary vectors of length $k$.
\end{definition}

\begin{definition}[Minwise Hash function for Jaccard
    similarity]\label{defn:minwise}
    Let $\pi$ be some permutations over $\{1, \ldots, n\}$.
    Treating $I$ as a subset of indices, we will use hash functions of the form
    $h_\pi(I) = \arg\min_i \pi(i)$ for $i \in I$.
\end{definition}
The probabilities that two itemsets hash to the same value for these hash
functions are related to their Hamming distance and Jaccard similarity, respectively.

\subsection{Apriori algorithm for frequent itemset mining}
As explained earlier, Apriori works in level-by-level, where the objective of
level-$l$ is to generate all $\theta$-frequent itemsets with $l$-items
each; for example, in the first level, the algorithm simply computes support of
individual items and retains the ones with support at least $\theta n$.
Apriori processes each level, say level-$(l+1)$, by joining all
pairs of $\theta$-frequent {\em compatible itemsets} generated in level-$l$,
and further filtering out the ones which have
support less than $\theta n$ (support computation involves fetching the
actual transactions from disk). Here, two candidate itemsets (of size $l$
each) are said to be compatible if their union has size exactly $l+1$. A
high-level pseudocode of Apriori is given in Algorithm \ref{algorithm:apriori}. 

\begin{algorithm}[H]\label{algorithm:apriori}
   \KwInput{Transaction database $\D$, support threshold $\theta$;}
    \KwResult{$\theta$-frequent itemsets;}
    $l = 1$ /* level */\;
    $F = \big\{ \{x\} ~|~ \{x\} \mbox{ is $\theta$-frequent in
	$\D$}\big\}$ /* frequent itemsets in level-1 */ \;
    Output $F$\;
    \While{$F$ is not empty}{
	$l = l+1$\;
	$C = \{ I_a \cup I_b ~|~ I_a \in F, ~I_b \in F, ~\mbox{$I_a$ and $I_b$
	    are compatible}	\}$\;\label{line:compute-candidates}
	$F = \emptyset$\;\label{line:begin-support-checking}
	\For{itemset $I$ in $C$}{
	    Add $I$ to $F$ if support of $I$ in $\D$ is at least $\theta n$ /* reads database*/\label{line:support-checking}
	    \;
	}\label{line:end-support-checking}
	Output $F$\;
    }
\caption{Apriori algorithm for frequent itemset mining}
\end{algorithm}

\section{LSH-Apriori}\label{sec:LSH-Apriori}
The focus of this paper is to reduce the computation of processing all pairs of
itemsets at each level in line \ref{line:compute-candidates} (which includes
computing support by going through $\D$).
Suppose that level $l$ outputs $m_l$ frequent itemsets.
We will treat the output of level $l$ as a collection of $m_l$
transaction vectors $\D_l = \{I_1, \ldots I_{m_l} \}$, each of length $n$ and one for
each frequent itemset of the $l$-th level.
Our approach involves defining appropriate
notions of similarity between itemsets (represented by
vectors) in $\D_l$ similar to the approach followed by Cohen {\it et
al.}\cite{Cohen01}. Let  $I_i, I_j$ be two vectors each  of length $n$. 
Then, we use $|I_i, I_j|$ to  denote the number of bit positions where both the vectors
have a $1$.

\begin{definition}\label{definition:similar}
Given a parameter $0<\varepsilon < 1$,  we say
that $\{I_i, I_j\}$  is $\theta$-frequent (or similar) if 
$|I_i, I_j|\geq \theta n$ and $\{I_i, I_j\}$ is $(1-\varepsilon)\theta$-infrequent if 
$|I_i, I_j|< (1-\varepsilon)\theta n$. Furthermore, we say that $I_j$ is {\em
similar} to $I_i$ if $\{I_i, I_j\}$ is $\theta$-frequent.
\end{definition}


Let $I_q$ be a frequent itemset at level $l-1$.
Let $\FI(I_q,\theta)$ be the set of itemsets $I_a$ such that $\{I_q, I_a\}$ is
$\theta$-frequent at level $l$. Our main contributions are a few randomized algorithms for
identifying itemsets in $\FI(I_q,\theta)$ with high-probability.


\begin{definition}[$\FI(I_q, \theta, \varepsilon, \delta)$]\label{definition:frequent} 
    Given a $\theta$-frequent itemset $I_q$ of size $l-1$, tolerance $\varepsilon \in (0,1)$ and error probability $\delta$,
    $\FI(I_q, \theta, \varepsilon, \delta)$ is a set $F'$ of itemsets of size
    $l$, such that with probability at least $1-\delta$, $F'$ contains {\em
    every} $I_a$ for which $\{I_q,I_a\}$ is $\theta$-frequent.
\end{definition}

It is clear that $\FI(I_q,\theta) \subseteq \FI(I_q, \theta,
\varepsilon, \delta)$ with high probability.
This motivated us to propose LSH-Apriori, a randomized version of Apriori, that takes
$\delta$ and $\varepsilon$ as additional inputs and essentially replaces line
{\bf \ref{line:compute-candidates}} by LSH operations to combine every itemset
$I_q$ with only similar itemsets, unlike Apriori which combines all pairs of itemsets.
This potentially creates a significantly smaller
$C$ without missing out too many frequent itemsets. 
The modifications to Apriori are
presented in Algorithm~\ref{algorithm:rand-apriori} and the following lemma,
immediate from Definition~\ref{definition:frequent}, establishes correctness of
LSH-Apriori.

\begin{algorithm}[H]\label{algorithm:rand-apriori}
    \LinesNumberedHidden
    \KwInput{$\D_l=\{I_1, \ldots, I_{m_l} \}$, $\theta$, (Additional) error probability $\delta$, tolerance $\varepsilon$;}
    \nlset{\ref{line:compute-candidates}a} (Pre-processing) Initialize
    hash tables and add all items $I_a \in \D_l$\;
    \nlset{\ref{line:compute-candidates}b} (Query) Compute $\FI(I_q,
    \theta, \varepsilon, \delta)$  $\forall I_q \in \D_l$ by hashing $I_q$ and
    checking collisions\;
    \nlset{\ref{line:compute-candidates}c} $C \leftarrow $ $\{ I_q \cup I_b ~|~
	I_q \in \D_l, ~I_b \in \FI(I_q, \theta, \varepsilon, \delta)\}$\;
	\caption{LSH-Apriori level $l+1$ (only modifications to Apriori line:\ref{line:compute-candidates})}
\end{algorithm}

\begin{lem} \label{lem:paper_main}
Let $I_q$ and $I_a$ be two $\theta$-frequent compatible itemsets of size
$(l-1)$ such that the itemset $J=I_q \cup I_a$ is also
$\theta$-frequent. Then, with probability at least $1-\delta$, $\FI(I_q,\theta,\varepsilon,\delta)$ contains $I_a$ (hence $C$
	contains $J$).
\end{lem}

In the next section we describe three LSH-based randomized algorithms to compute $\FI(I_q, \theta,
\varepsilon, \delta)$ for all $\theta$-frequent itemset $I_q$ from the earlier
level. The input to these subroutines will be $\D_l$, the frequent itemsets from
earlier level, and parameters $\theta, \varepsilon, \delta$.
In the {\em pre-processing stage} at level $l$, the respective LSH is
initialized and itemsets of $\D_l$ are hashed; we specifically record the
itemsets hashing to every bucket.
LSH guarantees (\textit{w.h.p.}) that pairs of similar items hash 
into the same bucket, and those that are not hash into 
different buckets. In the {\em query stage}
we find all the itemsets that any $I_q$ ought to be
combined with by looking in the bucket in which $I_q$ hashed, and then
combining the compatible ones among them with $I_q$ to form $C$. Rest of the
processing happens \`{a} la Apriori.

The internal LSH subroutines may output false-positives -- itemsets that are
not $\theta$-frequent, but such itemsets are eventualy filtered out in
line~\ref{line:support-checking} of Algorithm~\ref{algorithm:apriori}. Therefore, the output of
LSH-Apriori does not contain any false positives. However, some frequent itemsets may be missing from its
output (false negatives) with some probability
depending on the parameter $\delta$ as stated below in
Theorem~\ref{thm:lsh-apriori-error} (proof follows from the union bound
and is give in the Appendix). 
\begin{restatable}[Correctness]{thm}{thmlshapriorierror}\label{thm:lsh-apriori-error}
      LSH-Apriori does not output any itemset which is not $\theta$-infrequent.
    If $X$ is a $\theta$-frequent itemset of size $l$, then the probability that
    LSH-Apriori does not output $X$ is at most $\delta2^l$. 
    \end{restatable}

The tolerance parameter $\varepsilon$ can be used to balance the overhead from
using hashing in LSH-Apriori with respect to its savings because of reading
fewer transactions.
Most LSH, including those that we will be using,
behave somewhat like dimensionality reduction. As a result, the
hashing operations do not operate on all bits of the vectors. Furthermore, the
pre-condition of similarity for joining ensure that (w.h.p.) most infrequent itemsets can be
detected before verifying them from $\D$.
To formalize this, consider any level $l$ with $m_l$ $\theta$-frequent itemsets
$\D_l$.
We will compare the computation done by LSH-Apriori at level $l+1$ to what
Apriori would have done at level $l+1$ given the same frequent itemsets $\D_l$.
Let $c_{l+1}$ denote the number of
candidates Apriori would have generated and $m_{l+1}$ the number of frequent
itemsets at this level (LSH-Apriori may generate fewer).

\newcommand{\Space}{\sigma}
\newcommand{\overhead}{\vartheta}
\noindent{\bf Overhead:} 
Let $\tau(LSH)$ be the time required for hashing an
itemset for a particular LSH and let $\Space(LSH)$ be
the space needed for storing respective hash values.
The extra overhead in terms of space will be simply $m_l \Space(LSH)$ in level
$l+1$. With respect to overhead in running time, LSH-Apriori requires hashing
each of the $m_l$ itemsets twice, during pre-processing and during querying.
Thus total time overhead in this level is $\overhead(LSH,l+1)=2 m_l \tau(LSH)$.

\newcommand{\Rej}{\mathcal{R}}
\newcommand{\sav}{\varsigma}

\noindent{\bf Savings:}
Consider the itemsets in $\D_l$ that are compatible with any $I_q \in \D_l$. Among them are those whose combination with $I_q$ do not generate a
$\theta$-frequent itemset for level $l+1$; call them as {\em negative} itemsets
and denote their number by $r(I_q)$.
Apriori will have to read all $n$ transactions of $\sum_{I_q} r(I_q)$ itemsets
in order to reject them. 
Some of these negative itemsets will be added to $\FI$ by LSH-Apriori -- we will
call them {\em false positives} and denote their count  by $FP(I_q)$; the rest
those which correctly not added with $I_q$ -- lets call them   as {\em true negatives} and denote their count by $TN(I_q)$. 
Clearly, $r(I_q)=TN(I_q) +
FP(I_q)$ and $\sum_{I_q} r(I_q)=2(c_{l+1}-m_{l+1})$. Suppose $\phi(LSH)$ denotes the number of transactions a particular
LSH-Apriori reads for hashing any itemset; due to the  dimensionality reduction property of LSH, $\phi(LSH)$ is
always $o(n)$. Then, LSH-Apriori   is able to reject all itemsets in
$TN$ by reading only $\phi$ transactions for each of them; thus for itemset $I_q$ in
level $l+1$, a particular LSH-Apriori reads $(n-\phi(LSH))\times TN(I_q)$ fewer transactions 
compared to a similar situation for Apriori. Therefore, total savings at level
$l+1$ is ${\sav(LSH,l+1) = (n-\phi(LSH)) \times \sum_{I_q} TN(I_q)}$.

In Section~\ref{sec:FIviaLSH}, we  discuss
this in more detail along with the respective LSH-Apriori algorithms.

\section{\FI~via LSH}\label{sec:FIviaLSH}

Our similarity measure $|I_a,I_b|$ can also be seen as the inner
product of the binary vectors $I_a$ and $I_b$. 
 However,  
it is not possible to get any LSH for 
such similarity measure because for example there can be three 
items $I_a, I_b$ and $I_c$ such that $|I_a, I_b|\geq |I_c, I_c|$ which implies 
that $\Pr(h(I_a)=h(I_b))\geq \Pr(h(I_c)=h(I_c))=1 $, which is not possible. 
 Noting the exact same problem, Shrivastava {\it et al.}\ introduced the
concept of {\em asymmetric LSH}~\cite{Shrivastava015} in the context of binary
inner product similarity.
The essential idea is to use two different hash functions (for pre-processing
and for querying) and they specifically proposed extending MinHashing by
padding input vectors before hashing. We use the same pair of padding functions
proposed by them for $n$-length binary vectors in a level $l$: $P_{(n,\alpha_l)}$
for preprocessing and $Q_{(n,\alpha_l)}$ for querying are defined as follows.
\begin{itemize}
    \item In $P(I)$ we append $(\alpha_l n-|I|)$ many $1'$s followed by
	$(\alpha_l n+|I|)$ many $0'$s.
    \item In $Q(I)$ we append $\alpha_l n$ many $0'$s, then $(\alpha_l n-|I|)$ many $1'$s, then $
    |I|$ $0'$s.
\end{itemize}
Here, $\alpha_l n$ (at LSH-Apriori level $l$) will denote the
maximum number of ones in any itemset in $\D_l$. Therefore, we always have $(\alpha_l n -
|I|) \ge 0$ in the padding functions. Furthermore, since the main loop of
Apriori is not continued if no frequent itemset is generated at any level,
$(\alpha_l- \theta) > 0$ is also ensured at any level that Apriori is
executing.

We use the above padding functions to reduce our problem of finding
similar itemsets to finding nearby vectors under Hamming distance
(using Hamming-based LSH in Subsection~\ref{subsection:HammingLSH} and Covering
LSH in Subsection~\ref{subsection:coveringLSH}) and under Jaccard similarity (using
MinHashing in Subsection~\ref{subsection:minHashLSH}).

\subsection{Hamming based LSH}\label{subsection:HammingLSH}
In the following lemma (proof is given in appendix), we relate Hamming distance of two  itemsets $I_x$ and $I_y$
with their $|I_x, I_y|$.
\begin{restatable}{lem}{lemhamPadding}\label{lem:hamPadding}
For two itemsets $I_x$ and $I_y$, $\Ham(P(I_x), Q(I_y))=2(\alpha_l n-|I_x, I_y|).$
\end{restatable}

Therefore, it is possible to use an LSH for Hamming distance to find similar
itemsets. We use this technique in the following
algorithm to compute $\FI(I_q, \theta, \varepsilon, \delta)$ for all itemset
$I_q$. The algorithm contains an optimization over the generic
LSH-Apriori pseudocode (Algorithm~\ref{algorithm:rand-apriori}). There is no need to separately execute
lines:\ref{line:begin-support-checking}--\ref{line:end-support-checking}
of Apriori; one
can immediately set $F \leftarrow C$ since LSH-Apriori computes support before
populating $\FI$.
 
\begin{algorithm}[H]\label{algorithm:algoLSHviaHamming}
\LinesNumberedHidden
\KwInput
{$\D_l = \{I_1, \ldots, I_{m_l}\}$, query item $I_q$, threshold $\theta$, tolerance $\varepsilon$, error $\delta$.}  
\KwResult{$\FI_q = \FI(I_q, \theta, \varepsilon, \delta)$ for every $I_q \in \D_l$.}
\nlset{\ref{line:compute-candidates}a} \textbf{Preprocessing step:} Setup hash tables and add vectors in $\D_l$\;
    \Indp\nlset{i} Set $\rho=\frac{\alpha_l-\theta }{\alpha_l-(1-\varepsilon)\theta }$,
    $k=\log_{\left(\frac{1+2\alpha_l}{(1+2(1-\varepsilon)\theta)}\right)}m_l$
    and $L=m_l^{\rho}\log \left(\frac{1}{\delta}\right)$\;
    \nlset{ii} Select functions $g_1,\ldots,g_L$ \textit{u.a.r.}\;
    \nlset{iii} For every $I_a \in \D_l$, pad $I_a$ using $P()$ and then hash
	$P(I_a)$ into buckets $g_1(P(I_a)), . . . , g_L(P(I_a))$\;
\Indm\nlset{\ref{line:compute-candidates}b} \textbf{Query step:} For every $I_q \in \D_l$, we do the following \; 
    \Indp \nlset{i} $S \leftarrow $ all $I_q$-compatible itemsets in all buckets $g_i(Q(I_q))$, for $i=1 \ldots L$\;
    \nlset{ii} \For{$I_a \in S$}{
	If $|I_a, I_q| \geq  \theta n$, then add $I_a$ to $\FI_q$ /* reads database*/\;
	(*) If no itemset similar to $I_q$ found within $\frac{L}{\delta}$ tries,
	then break loop\;
    }
    \caption{LSH-Apriori (only lines \ref{line:compute-candidates}a,\ref{line:compute-candidates}b) using Hamming LSH}
\end{algorithm}

Correctness of this algorithm is straightforward. Also, $\rho < 1$ and
the space required and overhead of reading
transactions is $\theta(kLm_l)=o(m_l^2)$. It can be further shown that
$\E[FP(I_q)] \le L$ for $I_q \in \D_l$ which can be used to prove that $\E[\sav]
\ge (n-\phi)(2(c_{l+1}-m_{l+1})-m_l L)$ where $\phi=kL$. Details of these calculations including
complete proof of the next lemma is given in Appendix.
 
\begin{restatable}{lem}{lemoverheadHamming}\label{lem:klHamming}
Algorithm~\ref{algorithm:algoLSHviaHamming} correctly outputs 
$\FI(I_q, \theta, \varepsilon, \delta)$ for all $I_q \in \D_l$.
Additional space required is $o(m_l^{2})$, which is also
the total time overhead. 
The expected savings can be bounded by 
$\E[\sav(l+1)] \ge \big(n-o(m_l)\big)\big((c_{l+1}-2m_{l+1}) + (c_{l+1}-o(m_l^2))\big)$.
\end{restatable}

Expected savings outweigh time overhead if $n \gg m_l$, 
${c_{l+1} = \theta(m_l^2)}$ and $c_{l+1} > 2m_{l+1}$, i.e., in levels where the
number of frequent itemsets generated are fewer compared to the number of
transactions as well as to the number of candidates generated.
The additional optimisation 
(*) essentially increases the savings when all $l+1$-extensions of $I_q$ are
$(1-\varepsilon)\theta$-infrequent --- this behaviour will be predominant in the
last few levels. It is easy to show that 
in this case, $FP(I_q) \le \frac{L}{\delta}$ with probability at least
$1-\delta$; this in turn implies that $|S|\le \frac{L}{\delta}$. 
So, if we did not find any similar $I_a$ within first $\frac{L}{\delta}$ tries,
then we can be sure, with reasonable probability, that there are no itemsets
similar to $I_q$.

\subsection{Min-hashing based LSH}\label{subsection:minHashLSH}
Cohen et al.\ had given an LSH-based randomized algorithm for finding
interesting itemsets without any requirement for high support~\cite{Cohen01}.
We observed that their Minhashing-based technique~\cite{BroderCFM00} cannot be directly applied to the high-support
version that we are interested in. The reason is roughly that Jaccard similarity
and itemset similarity (w.r.t. $\theta$-frequent itemsets) are not monotonic to
each other. Therefore, we used padding to monotonically relate Jaccard similarity of two  itemsets $I_x$ and $I_y$
with their $|I_x, I_y|$ (proof is given in Appendix).
 \begin{restatable}{lem}{lemminhashPadding}\label{lem:minhashPadding}
For two padded itemsets $I_x$ and $I_y$, $ \JS(P(I_x), Q(I_y))=\frac{|I_x, I_y|}{2\alpha_l n-|I_x, I_y|}.$
\end{restatable}


\newcommand{\DHat}{\hat{D_l}}

Once padded, we follow similar steps (as \cite{Cohen01}) to create a {\em similarity
preserving summary} $\DHat$ of $\D_l$ such that the Jaccard similarity for
any column pair in $\D_l$ is approximately preserved in $\DHat$, and then
explicitly compute $\FI(I_q, \theta, \varepsilon, \delta)$ from $\DHat$.
$\DHat$ is created by using $\lambda$ independent minwise hashing functions (see
Definition~\ref{defn:minwise}). $\lambda$ should be carefully chosen since a higher value
increases the accuracy of estimation, but at the cost of \textit{large summary
vectors} in $\DHat$.
%
%
  Let us define $\hat{\JS}(I_i, I_j)$ 
  as the fraction of rows in the  summary matrix in which min-wise entries of columns $I_i$ and $I_j$
  are identical. Then by Theorem 1 of Cohen et al.~\cite{Cohen01}, we can get a bound on the 
  number of required hash functions:

\begin{thm}[Theorem $1$ of \cite{Cohen01}]\label{theorem:cohen}
 Let $0<\epsilon, \delta<1$ and $\lambda \geq \frac{2}{\omega\epsilon^2}\log{\frac{1}{\delta}}$. Then for all pairs of columns
 $I_i$ and $I_j$ following   are true with probability at least $1-\delta$:
 \begin{itemize}
  \item If $\JS(I_i, I_j)\geq s* \geq \omega$, then $\hat{\JS}(I_i, I_j)\geq(1-\epsilon)s*$, 
  \item If $\JS(I_i, I_j) \leq \omega$, then $\hat{\JS}(I_i, I_j)\leq(1+\epsilon)\omega$.
 \end{itemize}
\end{thm}

\footnote{This algorithm can be easily boosted to $o(\lambda m_l)$ time by applying banding technique 
(see Section $4$ of ~\cite{Cohen01}) on the minhash table. }
 \begin{algorithm}[H]\label{algorithm:algoLSHviaMinHash}
\LinesNumberedHidden
\KwInput
{$\D_l$, query item $I_q$, threshold $\theta$, tolerance $\varepsilon$, error $\delta$}  
\KwResult{$\FI_q = \FI(I_q, \theta, \varepsilon, \delta)$ for every $I_q \in \D_l$.}
\nlset{\ref{line:compute-candidates}a} \textbf{Preprocessing step:} Prepare
$\DHat$ via MinHashing\;
    \Indp\nlset{i} Set
    $\omega=\frac{(1-\varepsilon)\theta}{2\alpha_l-(1-\varepsilon)\theta }$, $\epsilon =
\frac{\alpha_l\varepsilon}{\alpha_l + (\alpha_l-\theta)(1-\varepsilon)}$ and
$\lambda=\frac{2}{\omega\epsilon^2}\log{\frac{1}{\delta}}$\;
    \nlset{ii}Choose $\lambda~$ many independent permutations (see Theorem~\ref{theorem:cohen})\;
    \nlset{iii} For every $I_a \in \D_l$, pad $I_a$ using $P()$ and then hash
	$P(I_a)$ using $\lambda$ independent permutations\;
\Indm\nlset{\ref{line:compute-candidates}b} \textbf{Query step:} For every $I_q \in \D_l$, we do the following \; 
    \Indp \nlset{i} Hash $ Q(I_q)$ using $\lambda$ independent permutations\;
    \nlset{ii} \For{compatible $I_a \in \D_l$}{
	If $\hat{\JS}(P(I_a), Q(I_q))\geq \frac{(1-\epsilon)\theta }{2\alpha_l-\theta}$ for some $I_a$, then add $I_a$ to $\FI_q$\;
    }
    \caption{LSH-Apriori (only lines \ref{line:compute-candidates}a,\ref{line:compute-candidates}b) using Minhash LSH}
\end{algorithm}
 \newcomment{
}
\begin{restatable}{lem}{lemMinhashlemma}\label{lem:Minhashlemma}
 Algorithm~\ref{algorithm:algoLSHviaMinHash} correctly computes $\FI(I_q,\theta, \varepsilon, \delta)$ 
 for all $I_q \in \D_l$.   Additional space required is $O(\lambda m_l) $, and 
  the total time overhead is $O((n+\lambda)m_l)$. The expected savings is 
given by  $\E[\sav(l+1)] \ge  2(1-\delta)(n-\lambda)(c_{l+1}-m_{l+1})$.
\end{restatable}


See Appendix for details of  the above proof. Note that $\lambda$ depends on
$\alpha_l$ but is independent of $n$. This method should be applied only when
$\lambda \ll n$. And in that case, for levels with number of candidates
much larger than the number of frequent itemsets discovered (i.e.,
$c_{l+1} \gg \{m_l, m_{l+1} \}$), time overhead would not appear
significant compared to expected savings.

\newcomment{
}

\subsection{Covering LSH}\label{subsection:coveringLSH}
Due to their probabilistic nature, the LSH-algorithms presented earlier
have the limitation of producing false positives and more importantly, false
negatives. Since the latter cannot be detected unlike the former, these
algorithms may miss some frequent itemsets (see
Theorem~\ref{thm:lsh-apriori-error}). In fact,
once we miss some  \FI~ at a particular level, then all the \FI~ which are
\textit{``supersets"} of that \FI~ (in the subsequent levels) will be missed.
Here we present another algorithm for the same purpose which overcomes this
drawback. 
The main tool is a recent algorithm due to Pagh~\cite{Pagh16} which returns
approximate nearest neighbors in the Hamming space. It is an improvement over the seminal LSH algorithm by Indyk
and Motwani~\cite{IndykM98}, also for Hamming distance. Pagh's algorithm has a
small overhead over the latter; to be precise, the query time bound of~\cite{Pagh16} differs 
 by at most $\ln (4)$ in the exponent in comparison with the time bound of~\cite{IndykM98}.
 However, its big advantage is that it generates no false negatives. Therefore,
 this LSH-Apriori version also does not miss any frequent itemset.


 The LSH by Pagh is with respect to Hamming distance, so we first reduce our
 \FI~problem into the Hamming space by using the same padding given in
 Lemma~\ref{lem:hamPadding}.
 Then we use this LSH in the same manner as in Subsection~\ref{subsection:HammingLSH}.
 Pagh coined his hashing scheme as coveringLSH which broadly mean that given a threshold $r$ 
 and a tolerance $c>1$, the  hashing scheme guaranteed a collision for every 
 pair of vectors that are within radius $r$.
 We will now briefly summarize coveringLSH for our requirement;
 refer to the paper~\cite{Pagh16} for full details.

 Similar to HammingLSH, we use a family of Hamming projections as our hash functions: 
 $\mathcal{H}_{\mathcal{A}}:=\{x\mapsto x\wedge a|~ a\in \mathcal{A} \}$, where $\mathcal{A}\subseteq \{0, 1\}^{(1+2\alpha_l)n}$.
  Now, given a query item $I_q$, the idea is to iterate through all hash functions  $h \in \mathcal{H}_{\mathcal{A}}$,  
 and check if there is a collision $h(P(I_x))=h(Q(I_q))$ for $I_x\in \mathcal{D}_l$. 
We say that this scheme doesn't produce false negative for the threshold $2(\alpha_l-\theta)n$, if at least one 
collision happens when there is an $I_x\in\mathcal{D}_l$ when $\Ham(P(I_x), Q(I_q))\leq 2(\alpha_l-\theta)n$, and 
the scheme is efficient if the number of  collision is not too many 
when $\Ham(P(I_x), Q(I_q))> 2(\alpha_l-(1-\varepsilon)\theta)n$ (proved in
Theorem $3.1, 4.1$ of~\cite{Pagh16}).
 To make sure that all pairs of vector within distance $2(\alpha_l-\theta)n$
 collide for some $h$, 
 we need to make sure that some $h$ map their ``mismatching'' bit
 positions (between $P(I_x)$ and $Q(I_q)$) to $0$. We describe construction of
 hash functions next. 

 \begin{tabular}{|c|c|c|c|p{3.5cm}|c|}
\hline
$n'$ & $\theta'$ &   $t$ & $c$ & \centering $\epsilon$   & $\nu$\\
\hline
${\scriptstyle (1+2\alpha_l)n}$ &
${\scriptstyle 2(\alpha_l-\theta)n}$ &
$\lceil\frac{\ln  m_l}{2(\alpha_l-(1-\varepsilon)\theta)n} \rceil$ &
$\frac{\alpha_l-(1-\varepsilon)\theta}{\alpha_l-\theta}$ &
\parbox{\textwidth}{$\epsilon\in(0, 1)$
\textit{s.t.}\\ $\frac{\ln  m_l}{2(\alpha_l-(1-\varepsilon)\theta)n}+\epsilon\in
\mathbb{N}$} &
$\frac{t+\epsilon}{ct}$ \\
\hline
\end{tabular}
 \newcomment{
}
\paragraph*{CoveringLSH:} The parameters relevant to LSH-Apriori are given above.
Notice that after padding, dimension
of each item  is $n'$,  threshold is $\theta'$ (i.e., min-support is
$\theta'/n'$), and  tolerance is $c$.
We start by choosing a random function $\varphi:\{1,\ldots, n'\} \rightarrow \{0,
1\}^{t\theta'+1}$, which 
maps  bit positions of the padded itemsets to bit vectors of length
$t\theta'+1$. 
We define a family of bit vectors $a(v)\in \{0, 1\}^{n'}$, where
$a(v)_i=\langle \varphi(i), v\rangle$, for $i\in\{1, \ldots , n'\}$, $v\in \{0,
1\}^{t\theta'+1}$  
and $\langle m(i), v\rangle$ denotes the inner product over $\mathbb{F}_2$. We
define our hash function family ${\cal H}_{\cal A}$ using all such vectors
$a(v)$ except $a(\mathbf{0})$: $\mathcal{A} =\left\{a(v)|v\in \{0,
1\}^{t\theta'+1}/ \{\textbf{0}\}  \right\}$.

\newcommand{\A}{\mathcal{A}}

Pagh described how to construct $\A' \subseteq \A$~\cite[Corollary 4.1]{Pagh16}
such that ${\cal H}_{\A'}$ has a very useful property of no false negatives and
also ensuring very few false positives. We use ${\cal H}_{\A'}$ for hashing
using the same manner of Hamming projections as used in
Subsection~\ref{subsection:HammingLSH}. Let $\psi$ be the
expected number of collisions between any itemset $I_q$ and items in
$\mathcal{D}_l$ that are ${(1-\varepsilon)\theta}$-infrequent with $I_q$.
The following Theorem captures the essential property of coveringLSH that is
relevant for LSH-Apriori, described in Algorithm~\ref{algorithm:coveringLSH}. 
It also bounds the number of hash functions which controls the space and time overhead of LSH-Apriori.
Proof of this theorem follows from Theorem $4.1$ and Corollary $4.1$ of~\cite{Pagh16}.

\begin{thm}\label{thm:coveringLSHthm}
For a randomly chosen $\varphi$, a hash family $\mathcal{H}_{A'}$ described above and distinct $I_x, I_q \in \{0, 1\}^n:$
\begin{itemize}
  \item If $\Ham\big(P(I_x), Q(I_q)\big)\leq \theta'$, then there exists $h \in
      {\cal H}_{\A'}$ s.t. ${h\big(P(I_x)\big)=h\big(Q(I_q)\big)},$
  \item Expected number of false positives is bounded by $\E[\psi] < 2^{{\theta'}\epsilon+1}m_l^{\frac{1}{c}}$, 
  \item $|\mathcal{H}_{A'}|< 2^{{\theta'}\epsilon+1}m_l^{\frac{1}{c}}.$
\end{itemize}
\end{thm}

 \begin{algorithm}[H]\label{algorithm:coveringLSH}
\LinesNumberedHidden
\KwInput
{$\D_l$, query item $I_q$, threshold $\theta$, tolerance $\varepsilon$, error $\delta$.}  
\KwResult{$\FI_q = \FI(I_q, \theta, \varepsilon, \delta)$ for every $I_q \in \D_l$.}
\nlset{\ref{line:compute-candidates}a} \textbf{Preprocessing step:} 
Setup hash tables according to $\mathcal{H}_{\mathcal{A}'}$ and add items\;
    \Indp\nlset{i} For every $I_a \in \D_l$, hash
	 $P(I_a)$ using all $h \in \mathcal{H}_{\mathcal{A}'}$\;
\Indm\nlset{\ref{line:compute-candidates}b} \textbf{Query step:} For every $I_q \in \D_l$, we do the following \; 
    \Indp \nlset{i} $S \leftarrow $ all itemsets that collide  with $Q(I_q)$\;
    \nlset{ii} \For{$I_a \in S$}{
	If $|I_a , I_q | \geq \theta n$, then add $I_a$ to $\FI_q$ /* reads database*/\;
	(*) If no itemset similar to $I_q$ found within $\frac{\psi}{\delta}$ tries,
	break loop\;
    }
    \caption{LSH-Apriori (only lines \ref{line:compute-candidates}a,\ref{line:compute-candidates}b) using Covering LSH}
\end{algorithm}

 \begin{restatable}{lem}{lemcoveringlemma}\label{lem:coveringlemma}
   Algorithm~\ref{algorithm:coveringLSH} outputs all $\theta$-frequent itemsets     and only $\theta$-frequent itemsets.
   Additional space required is $O\left( m_l^{1+\nu}\right) $, 
  which is also the total time overhead.   The expected savings is 
given by  $\E[\sav(l+1)] \ge  2\left(n-\frac{\log m_l}{c}-1\right)\linebreak \left((c_{l+1}- m_{l+1}) - 
m_l^{1+\nu} \right)$.
\end{restatable}

See Appendix for the proof. The (*) line is an additional optimisation similar
to what we did for HammingLSH~\ref{subsection:HammingLSH}; it efficiently
recognizes those frequent itemsets of the earlier level none of whose extensions are
frequent.
The guarantee of not missing any valid itemset comes with a heavy price. Unlike
the previous algorithms, the conditions under which expected savings beats
overhead are quite stringent, namely, $c_{l+1} \in \{ \omega(m_l^2),
\omega(m_{l+1}^2) \}$, $\frac{2^{n}}{5} > m_l > 2^{n/2}$ and
$\epsilon < 0.25$ (since $1 < c < 2$, these
bounds ensure that $\nu < 1$ for later levels when $\alpha_l \approx \theta$).

\section{Conclusion}\label{sec:conclusion}

In this work, we  designed  randomized algorithms using
locality-sensitive hashing (LSH) techniques which efficiently outputs almost all the frequent itemsets 
with high probability at the cost of a little space which is required for
creating hash tables. We showed that time overhead is usually small compared to the
savings we get by using LSH.

Our work opens the possibilities for addressing a wide
range of problems that employ on various versions of frequent itemset and sequential pattern
mining problems, which potentially can efficiently be randomized using LSH techniques.

\bibliographystyle{abbrv}
\bibliography{reference}
\newpage
\noindent{\LARGE \textbf{Appendix}}

\thmlshapriorierror*
\begin{proof}
 LSH-Apriori does not output $X$, whose size we denote by $l$, if at least
one of these hold.
\begin{itemize}
 \item Any $1$ size subset of $X$ is not generated by LSH-Apriori in level-$1$
 \item Any $2$ size subset of $X$ is not generated by LSH-Apriori in level-$2$\\
 \vdots
 \item Any $l$ size subset of $X$ (i.e., $X$ itself) is not generated in level-$l$
\end{itemize}

 By Lemma~\ref{lem:paper_main}, $\delta$ is the probability that any particular frequent
itemset is not generated at the right level, even though all its
subsets were identified as frequent in earlier level. Since there are
${l \choose k}$ subsets of $X$ of size $k$, the required probability can be
upper bounded using Union Bound to
$${l \choose 1}\delta + {l \choose 2}\delta + ... + {l \choose l}\delta \le 2^l \delta.$$
\end{proof}

To get the the necessary background,  Lemma~\ref{lem:LSH} provide bounds on the 
hashing parameters $k, L$ for Hamming distance case. Their proof is adapted 
from ~\cite{GionisIM99,IndykM98,Cohen01}. We first require Lemma~\ref{lem:freBbound2}, \ref{lem:unFreBbound2} for the same.
\begin{lem} \label{lem:freBbound2}
 Let  $\{I_i, I_j\}$ be a pair of items \textit{s.t.}  $\Ham(I_i, I_j)\leq r$, 
  then the probability that 
$I_i$ and $I_j$ hash  into at least one of the  $L$ bucket of size $k$, is at least
$ 1-(1-{p_1}^k)^L$, where $p_1=1-\frac{r}{n}.$
\end{lem}
\begin{proof}
Probability that $I_i$ and $I_j$ matches at some particular bit position   $\geq  {p_1}$. Now, 
probability that $I_i$ and $I_j$ matches  at $k$ positions in a bucket of size 
$k$  $\geq {p_1}^k$.
Probability that $I_i$ and $I_j$ don't matches   at  $k$ positions in a bucket of size $k \leq 1-{p_1}^k $.
Probability that  $I_i$ and $I_j$ don't matches  at $k$   positions in none of the $L$ buckets  $\leq (1-{p_1}^k)^L $.
Probability that  $I_i$ and $I_j$  matches  in at $k$ positions  positions in at least one of the $L$ buckets $\geq 1-(1-{p_1}^k)^L $. 
 \end{proof}

\begin{lem}\label{lem:unFreBbound2}
 Let  $\{I_i, I_j\}$ be a pair of items \textit{s.t.} $\Ham(I_i, I_j)\geq (1+\epsilon')r$, then
 probability that    $\{I_i, I_j\}$ hash in a bucket of size $k$, is at most  ${p_2}^k$, where $p_2=1-\frac{(1+\epsilon')r}{n}.$
\end{lem}
 
\begin{proof}
Probability that $I_i$ and $I_j$ matches $1$ at some particular bit position   $<  p_2.$
Probability that $I_i$ and $I_j$ matches   at $k$ positions in a bucket of size 
$k  <  {p_2}^k$.
\end{proof}

\begin{lem}\label{lem:LSH}
 Let  $\{I_i\}_{i=1}^m$ be a set of $m$ vectors in $R^n$,  $I_q$ be a given query vector, and  $I_{x^*}~
 (\mbox{with,~}1\leq x^*\leq m)$  \textit{s.t.}  $\Ham(I_{x^*}, I_q)\leq r$. 
 If we set our hashing parameters 
 $k=\log_{\frac{1}{{p_2}}}m$, and $L=m^{\rho}\log\left(\frac{1}{\delta}\right)$ $(\mbox{where,~}p_1=1-\frac{r}{n}, p_2=1-\frac{r(1+\epsilon')}{n}, 
 \break  \rho= \frac{\log\frac{1}{{p_1}}}{\log \frac{1}{{p_2}}}\leq \frac{1}{1+\epsilon'})$, 
 then the following two cases are true  with  probability $> 1-\delta:$
 \begin{enumerate}
   \item for some $i\in \{1, . . . , L\}$, $g_i(I_{x^*} ) = g_i (I_q)$; and
  \item  total number of collisions with   $I_{x'}$  \textit{s.t.}  $\Ham(I_{x'}, I_q)>(1+\epsilon')r$ 
  is at most    $\frac{L}{\delta}$.
 \end{enumerate}
\end{lem}
 
\begin{proof}
Consider the first case, by Lemma~\ref{lem:freBbound2}, we have the following:\\
    $$\Pr[\exists i: g_i (I_{x^*} ) = g_i(I_q)] \geq 1-(1-{p_1}^k)^L.$$
If we choose $k=\log_{\frac{1}{{p_2}}}m$, we get
${p_1}^k={p_1}^{\log_{\frac{1}{{p_2}}}m}=m^{-\frac{\log\frac{1}{{p_1}}}{\log \frac{1}{{p_2}}}}$.
Let us denote $\rho= \frac{\log\frac{1}{{p_1}}}{\log \frac{1}{{p_2}}}$.
Then,   $\Pr[\exists i: g_i(I_{x^*} ) = g_i (I_q)] \geq  1-(1-m^{-\rho})^L.$ 
Now, if we set $L=m^{\rho}\log\left(\frac{1}{\delta}\right)$, then the required probability is 
$1-(1-m^{-\rho})^{m^{\rho}\log\left(\frac{1}{\delta}\right)} \geq 1-{\frac{1}{e}}^{\log\left(\frac{1}{\delta}\right)} > 1-\delta$.

Now, let us consider the case $2$. Let  $I_{x'}$ be an item such that  
$\Ham(I_q, I_{x'})>  r(1+\epsilon')$. Then by Lemma~\ref{lem:unFreBbound2}, we have the following:
$$\Pr[g_i(I_q)=g_i(I_{x'})] \leq {p_2}^k={p_2}^{\log_{{\frac{1}{{p_2}}}}m}=\frac{1}{m}  ~~(\mbox{as we choosed}~ k=\log_{{\frac{1}{{p_2}}}}m). $$ 
Thus, the expected number of collisions for a particular $i$ is at most $1$, and the expected 
total number of collisions is at most $L$ (by linearity of expectation).
Now, by Markov's inequality 
$\Pr[\mbox{Number of~} I_{x'} \mbox{~ which are colliding with~} I_q >\frac{L}{\delta} ] < \frac{L}{\left(\frac{L}{\delta}\right)}=\delta$.
Further, $\rho= \frac{\log\frac{1}{{p_1}}}{\log \frac{1}{{p_2}}}=\frac{n-r}{n-(1+\epsilon')r}\leq \frac{1}{1+\epsilon'}~\mbox{~(after simplification)}.$
\end{proof}
\lemhamPadding*
\begin{proof}
It is easy to verify that with this mapping $|P(I_x), Q(I_y)|=|I_x, I_y|$. 
Let $\Ham(I_x, I_y)$ denote the hamming distance between items $I_x$ and $I_y$. Then, 
 $ \Ham(P(I_x), Q(I_y))=|P(I_x)|+|Q(I_y)|-2|P(I_x), Q(I_y)|= \alpha_l n-|I_x|+|I_x|+\alpha_l n-|I_y|+|I_y|-2|I_x, I_y|
 =2(\alpha_l n-|I_x, I_y|).$
\end{proof}

\lemoverheadHamming*
\begin{proof}
First, we show that for any query item $I_q$,
 Algorithm~\ref{algorithm:algoLSHviaHamming} correctly outputs $\FI(I_q,
 \theta,  \varepsilon, \delta)$ for any query $I_q \in \D_l$.
 Now, if there is an item $I_{x^*}$ such that $|I_{x^*}, I_q|\geq \theta n$, then $\Ham(P(I_{x^*}), Q(I_q)) \leq 2(\alpha_l -\theta)n$ (by Lemma~\ref{lem:hamPadding}). 
 Let $p_1$ be the probability that $P(I_{x^*})$ and  $Q(I_q)$ matches at some particular bit position, then 
 $p_1\geq 1-\frac{2(\alpha_l n-\theta n)}{n+2\alpha_l n}=\frac{1+2\theta}{1+2\alpha_l }.$
 Similarly,  if there is an item $I_{x'}$ such that $|I_{x'}, I_q|\leq (1-\varepsilon)\theta n$, 
 then $\Ham(P(I_{x'}), Q(I_q))  \geq 2(\alpha_l -(1-\varepsilon)\theta)n$. Let $p_2$ be the probability that 
 $P(I_{x'})$ and  $Q(I_q)$ matches at some particular bit position, then 
 $p_2\leq 1-\frac{2(\alpha_l -(1-\varepsilon)\theta)n}{(1+2\alpha_l) n}=\frac{1+2(1-\varepsilon)\theta}{1+2\alpha_l }.$
 
 Now, as we have set $k=\log_{\frac{1}{p_2}}=\log_{\left(\frac{1+2\alpha_l }{1+2(1-\varepsilon)\theta}\right)}m_l;$
 and $L=m_l^{\rho}\log{\frac{1}{\delta}}, \mbox{~where}~ \rho=\frac{\log\frac{1}{{p_1}}}{\log \frac{1}{{p_2}}}=\frac{\log\frac{1+2\theta}{{1+2\alpha_l }}}{\log \frac{1+2(1-\varepsilon)\theta}{{1+2\alpha_l }}}
 =\frac{\alpha_l-\theta }{\alpha_l-(1-\varepsilon)\theta }~\mbox{~(after simplification)};$
 the proof easily follows from Lemma~\ref{lem:LSH}.

 The  space required for hashing an itemset $I_q$ is 
$\Space= O(kL)=\tilde{O}\left( {m_l}^{\rho}  \right)=m_l^{o(1)}$, 
where $\rho=\frac{\alpha_l-\theta }{\alpha_l-(1-\varepsilon)\theta } = o(1)$. 
Time $\tau$ require for hashing  $I_i$ is also
$O(kL)=\tilde{O}\left( {m_l}^{\rho}  \right)=m_l^{o(1)}$. 
Thus, total time and space overhead is $ {m_l}^{1+o(1)}$, which
 immediately proves the required space and time overhead.
 
The number of bits of any vector required by this LSH-function is
$\phi = O(kL) = {m_l}^{o(1)}$.
We know that $TN$ and $FP$ for an itemset $I_q \in \D_l$ are related by $TN(I_q) = r(I_q)
- FP(I_q)$, and  $\sum_{I_q} r(I_q) = 2(c_{l+1}-m_{l+1})$. So, $\E[\sum TN(I_q)] = \sum r(I_q) -
\E[\sum FP(I_q)]$. From Lemma~\ref{lem:LSH},
$\E[FP(I_q)] \le L$. Combining these facts, we get $\E[\sum TN(I_q)]
\ge 2(c_{l+1}-m_{l+1}) - m_l L= 2(c_{l+1}-m_{l+1}) - m_l m_l^{\rho}=2(c_{l+1}-m_{l+1}) -  m_l^{1+o(1)}$.
Now, using the formula for expected savings from Section~\ref{sec:LSH-Apriori},
\begin{align*}
 \E[\sav(l+1)] &\ge (n-O(kL))\left(2(c_{l+1}-m_{l+1}) -  m_l^{1+o(1)}\right)\\
 &=\left(n-  {m_l}^{o(1)} \right) (2(c_{l+1}-m_{l+1}) -  m_l^{1+o(1)} )\\ 
&\ge \big(n-o(m_l)\big)\big((c_{l+1}-2m_{l+1}) + (c_{l+1}-o(m_l^2))\big).
\end{align*}

\end{proof}

\lemminhashPadding*
\begin{proof}
The Jaccard Similarity between items $P(I_x)$  and $Q(I_y)$ is as follows:
 $ \JS(P(I_x), Q(I_y))=\frac{|P(I_x)\cap Q(I_y)|}{|P(I_x)\cup Q(I_y)|}=\frac{|P(I_x), Q(I_y)|}{|P(I_x)|+|Q(I_y)|-|P(I_x), Q(I_y)|}\\
  =\frac{|I_x, I_y|}{\alpha_l n-|I_x|+|I_x|+\alpha_l n-|I_y|+|I_y|-|I_x, I_y|} =\frac{|I_x, I_y|}{2\alpha_l n-|I_x, I_y|}.$
\end{proof}
 
\lemMinhashlemma*
 \begin{proof}
 Now, if there is an item $I_{x^*}$ such that $|I_{x^*}, I_q|\geq \theta n$, 
 then $\JS(P(I_{x^*}), Q(I_q))\geq \frac{\theta n }{(2\alpha_l-\theta) n}$ (by Lemma~\ref{lem:minhashPadding}).  
 As we  set $\omega=\frac{(1-\varepsilon)\theta  }{2\alpha_l-(1-\varepsilon)\theta }$ in Algorithm~\ref{algorithm:algoLSHviaMinHash},
 then  by Theorem~\ref{theorem:cohen} we have
 $\hat{\JS}(P(I_{x^*}), Q(I_q)) \geq \frac{(1-\epsilon)\theta n }{(2\alpha_l-\theta) n}$, with probability at least $1-\delta$. 
 
 Similarly, if there is an item $I_{x'}$ such that $|I_{x'}, I_q|< (1-\varepsilon)\theta n$, 
 then \newline $\JS(P(I_{x'}), Q(I_q))< \frac{(1-\varepsilon)\theta n }{2(\alpha_l-(1-\varepsilon)\theta) n}=\frac{(1-\varepsilon)\theta  }{2\alpha_l-(1-\varepsilon)\theta }$ 
 (by Lemma~\ref{lem:minhashPadding}).  
 Then by Theorem~\ref{theorem:cohen}, we have
 $\hat{\JS}(P(I_{x'}), Q(I_q)) < 
 \frac{(1+\epsilon)(1-\varepsilon)\theta}{2\alpha_l-(1-\varepsilon)\theta }$, with 
 probability at least $1-\delta$.
 
 We need to set Minhash parameter $\epsilon$ such that $\hat{\JS}(P(I_a), Q(I_q))\geq \frac{(1-\epsilon)\theta }{2\alpha_l-\theta}
 \geq \frac{(1+\epsilon)(1-\varepsilon)\theta  }{2\alpha_l-(1-\varepsilon)\theta}$. 
 This gives, $\epsilon< \frac{\alpha_l\varepsilon}{\alpha_l + (\alpha_l-\theta)(1-\varepsilon)}$,
 which ensures $|I_q, I_a|\geq(1-\varepsilon)\theta n$ with probability at least $1-\delta.$
 
  The  space required for hashing an itemset $I_i$ $(\mbox{for~} 1\leq i\leq m_l)$ is 
$\Space=O(\lambda)$. Where, $\lambda \geq \frac{2}{\omega\epsilon^2}\log{\frac{1}{\delta}}$ and $\epsilon =
\frac{\alpha_l\varepsilon}{\alpha_l + (\alpha_l-\theta)(1-\varepsilon)}$.
Total space   required for storing hash table is $O(m_l\lambda).$
 Creating hash table require one pass over $\D_l$, then preprocessing time overhead is $O(nm_l)$. 
We perform query on the hash table,  query time overhead is $O(\lambda m_l).$ Thus, total time overhead 
$\overhead(l+1)=O((n+\lambda)m_l).$
 
Now, if an itemset $I_a$ is infrequent with $I_q$, then $\Pr[I_a \mbox{~is not reported}] \geq 1-\delta.$
As there are $(c_{l+1}-m_{l+1})$ number of infrequent itemsets at level $l+1$, 
then, expected number of  infrequent items that are not reported $\E[TN] \geq (1-\delta)(c_{l+1}-m_{l+1}).$
Therefore, $\E[\sav(l+1)] = (n-\lambda)\E[TN] \geq (n-\lambda)(1-\delta)(c_{l+1}-m_{l+1}).$
 
\end{proof}
 
\newcomment{
 }
 \lemcoveringlemma*
 \begin{proof}
By Theorem~\ref{thm:coveringLSHthm}, and our choice of hash function $\mathcal{H}_{A'}$, 
any pair of similar itemset will surely collide by our hash function, and will not get missed by the algorithm. 
Moreover,   {\em false positives} will be filter out by the algorithm in Line~\ref{line:compute-candidates}b of  
Algorithm~\ref{algorithm:coveringLSH}. Thus, our algorithm outputs all $\theta$-frequent itemsets     and only $\theta$-frequent itemsets.

The  space required for hashing an itemset $I_q$ is 
$\Space= |\mathcal{H}_{A'}|$. 
Total space required for creating hash table is $O(m_l|\mathcal{H}_{A'}|)
=O\left(m_l  2^{{\theta'}\epsilon+1}m_l^{\frac{1}{c}} \right)=
O\left(m_l  m_l^{\frac{\epsilon}{ct}}m_l^{\frac{1}{c}}  \right)\\
=O\left(  m_l^{1+\frac{t+\epsilon}{ct}}\right)=O(m_l^{1+\nu})$.  
 Time $\tau$ require for hashing  $I_i$ is also $|\mathcal{H}_{A'}|$. 
Thus, total time overhead required (including both preprocessing and querying) is $\overhead(l+1)=O(m_l|\mathcal{H}_{A'}|)=
O(m_l^{1+\nu})$,  which   proves the required space and time overhead.
 
The number of bits of any item required by  our hash function is
$\phi =  \frac{\log m_l}{c}+1$.
We know that $TN$ and $FP$ for an itemset $I_q \in \D_l$ are related by $TN(I_q) = r(I_q)
- FP(I_q)$, and  $\sum_{I_q} r(I_q) = 2(c_{l+1}-m_{l+1})$. So, $\E[\sum TN(I_q)] = \sum r(I_q) -
\E[\sum FP(I_q)]$. From Theorem~\ref{thm:coveringLSHthm},
$\E[FP(I_q)] \le \psi$. Then, we get $\E[\sum TN(I_q)]
\ge 2(c_{l+1}-m_{l+1}) - m_l \psi$. 
 
 As  expected savings is  ${\sav(LSH,l+1) = (n-\phi) \times \sum_{I_q} TN(I_q)}$. We have, 
\begin{align*}
 \E[\sav(l+1)] &\ge \left(n-\frac{\log m_l}{c}-1\right)\left(2(c_{l+1}-m_{l+1}) -  m_l \psi\right).\\
   &\ge \left(n-\frac{\log m_l}{c}-1\right)\left(2(c_{l+1}-m_{l+1}) -  m_l 2^{{\theta'}\epsilon+1}m_l^{\frac{1}{c}} \right).\\
&\ge 2\left(n-\frac{\log m_l}{c}-1\right)\left(2(c_{l+1}-m_{l+1}) -   m_l^{1+\frac{t+\epsilon}{ct}}\right).\\
&= 2\left(n-\frac{\log m_l}{c}-1\right)\left(2(c_{l+1}-m_{l+1}) -   (m_l^{1+\nu})\right).\\
\end{align*}

We end with a quick proof that the following are sufficient to ensure
overhead is less than expected savings. $c_{l+1} \in \{ \omega(m_l^2),
\omega(m_{l+1}^2) \}$, $\frac{2^{n}}{5} > m_l > 2^{n/2}$, 
$\epsilon < 0.25$ and $\alpha \approx \theta$. Note that, $2^{cn} > 2^n$ and
$5m_l > 2^{c\cdot c^{1/d}}m_l$ for any $d > 0$. These imply,
\begin{align*}
    nc & > \log(m_l) + c\cdot c^{1/d}\\
    n - c^{1/d} & > \frac{\log(m_l)}{c} \\
    n - 1 &> \frac{\log(m_l)}{c}~~~~\mbox{since, $c^{1/d} > 1$}
\end{align*}

Furthermore, our conditions imply that $4 \epsilon (\alpha_l - \theta) c <
c-1$. This implies, 
\begin{align*}
    \frac{2\epsilon (\alpha_l - \theta) c}{c-1} & < 1/2\\
    n-2 > n/2 &> \frac{2\epsilon (\alpha_l - \theta) cn}{c-1}\\
    \log(m) &> \frac{2\epsilon (\alpha_l - \theta) cn}{c-1}\\
    1 = \frac{1}{c} + \frac{c-1}{c} &> \frac{1}{c} + \frac{2\epsilon (\alpha_l -
\theta) n}{\log(m)} = \frac{1}{c} + \frac{\epsilon}{ct} = \nu
\end{align*}

 \end{proof}

\end{document}